\documentclass[onefignum,onetabnum]{siamonline220329}

\usepackage{amsmath,amsfonts,amssymb}
\usepackage{mathtools}
\usepackage{color}
\usepackage{fullpage}
\usepackage{hyperref}
\usepackage{bm}
\usepackage{enumitem}
\setlist[enumerate]{leftmargin=.5in}
\setlist[itemize]{leftmargin=.5in}

\newsiamremark{remark}{Remark}
\newsiamremark{hypothesis}{Hypothesis}
\crefname{hypothesis}{Hypothesis}{Hypotheses}
\newsiamthm{claim}{Claim}

\title{Almost-lossless compression of a low-rank random tensor}
\author{Minh Thanh Vu\thanks{Huawei Research Center, Sweden.}}

\begin{document}

\maketitle

\begin{abstract}
In this work, we establish an asymptotic limit of almost-lossless compression of a random, finite alphabet tensor which admits a low-rank canonical polyadic decomposition.
\end{abstract}



\section{Introduction}
Tensors, or multiway arrays, have been emerging as powerful tools with applications in signal processing, computer vision, and big-data analytics \cite{cichocki2015tensor,kolda2009tensor}. In these applications tensors of data are approximated using products of lower dimensional structures like factor matrices, core tensors, etc. Different models have been proposed for the approximation purpose such as canonical polyadic decomposition (CPD) \cite{carroll1970analysis,harshman1970foundations}, Tucker's decomposition \cite{tucker1963}, etc. In addition to being an effective analytic tool, tensors also arise naturally from data sources such as color images, videos, etc.

An interesting research question would be how to compress and store multiway arrays of data efficiently. In this paper we put forward this research question by studying the asymptotic limit of tensor compression. For a tractable analysis, we assume in our study that the tensor can be factorized exactly as a sum of multiple rank-one components or a CPD model. This simple assumption is justifiable in practice since many tensors can be well-approximated by a few components. 
We then assume a probabilistic model for each factor matrix in which elements are drawn from a distribution on a finite alphabet.

In information theory, data are usually modeled in forms of vectors or sequences. Our model hence can be viewed as a generalization of previous data compression models. An achievable tensor compression-reconstruction scheme can be designed based on the conventional typical arguments. Due to the inherent dependence among elements inside a tensor, we need to use novel arguments for the converse proof. Since each factor matrix is randomly generated from a finite alphabet, it is full rank with high probability. We then use the connection between full rank factor matrices and the essential uniqueness of CP decomposition to establish the converse proof.

Our paper is organized as follows. In Section \ref{sec_2} we present assumptions used in our study and recap some definitions from multilinear algebra and information theory. We establish the asymptotic limit of compressing a random rank-one tensor in Section \ref{sec_3}. Then we provide two examples which highlight challenges of multiple-component scenarios as well as propose a way to tackle these problems in Section \ref{sec_4}. Finally, based on the observations in Section \ref{sec_4} we establish the asymptotic limit of almost-lossless compression of a random low-rank tensor in Section \ref{sec_5}.

\section{Preliminaries}\label{sec_2}
A tensor $\mathbf{T} \in \mathbb{R}^{I_1\times I_2\times\dots\times I_N}$ is a multiway array indexed by a tuple $(i_1,\dots,i_N)$ where $i_j\in [1:I_j]$. The parameter $N$ is called the order of the tensor. We say that $\mathbf{T}$ is a single component tensor if there exist vectors $\mathbf{a}_i\in\mathbb{R}^{I_i\times 1}$, $i\in [1:N]$, such that
\begin{displaymath}
{T}_{i_1,\dots,i_N} = a_{1i_1}\cdots a_{Ni_N}.
\end{displaymath}
In other words, $\mathbf{T}$ is the outer product of $\mathbf{a}_i$ and we write $\mathbf{T} = \mathbf{a}_1\circ\dots\circ \mathbf{a}_N$. For such a tensor we say that its rank is one if it is unequal to the all 0 tensor. Assume that $\mathbf{T}$ can be factorized into $R$ single component tensors $\mathbf{T}_i$, i.e.,
\begin{equation}
\mathbf{T} = \sum_{i=1}^R\mathbf{T}_i = \sum_{i=1}^R \mathbf{a}_{i1}\circ\dots\circ \mathbf{a}_{iN}. \label{factor_decomposi}
\end{equation}
Note that the arithmetic in our study is carried out on $\mathbb{R}$. We also write $\mathbf{T} = [\mathbf{X}_1;\cdots;\mathbf{X}_N]$ where each factor matrix is given by $\mathbf{X}_i = [\mathbf{a}_{1i},\dots,\mathbf{a}_{Ri}]$. The minimum of such $R$ in such factorizations of $\mathbf{T}$ is called the rank of $\mathbf{T}$. Determining the rank of a tensor is challenging, in fact NP-hard when $N\geq 3$  \cite{hillar2013most}. In this paper we assume that \textit{each realization of our random tensor $\mathbf{T}$ can be decompositions as in \eqref{factor_decomposi} with $R$ single components}, i.e., it does not necessarily mean that $R$ is the rank of the corresponding tensor. In our study we further assume that the parameter $R$ fixed.\\
To put our study into the information theoretic framework a probabilistic postulate of our data needs to be made. For this purpose we assume that factor matrices are independent unless otherwise stated. In each factor matrix $\mathbf{X}_i$, $i\in [1:N]$, the $(j,r)$-th entry $X_{i,jr}$ is drawn from a distribution $P_{\mathcal{X}_i,r}$ on an alphabet $\mathcal{X}_i$, $X_{i,jr}\sim P_{\mathcal{X}_i,r}$, for all $j\in [1:I_i]$ and $r\in [1:R]$. We further assume that the alphabets $\mathcal{X}_i$ are finite  and $\max_{i,r}\max_{x\in\mathcal{X}_i}P_{\mathcal{X}_i,r}(x)<1$. With this our main focus would be the set of tensors taking values in the alphabet $\mathcal{T}\triangleq\bigtimes_{i}\mathcal{X}_i$. For simplicity we only consider the case that all sizes are equal to each other
\begin{displaymath}
I_1=\dots = I_N = n,\;n\to\infty.
\end{displaymath}
 Now we are ready to define information theoretic quantities of interest.
\begin{definition}
For a given $n$, a tensor compression-reconstruction scheme consists of two mappings:
\begin{itemize}
\item a compression mapping $\phi_n\colon\mathcal{T}\to \mathcal{M}$ which maps a tensor $\mathbf{T}$ to an index $m\in \mathcal{M}$, namely $\phi_n(\mathbf{T}) = m$, which is stored in a storage medium such as a hard disk,
\item and a reconstruction mapping $\psi_n\colon \mathcal{M }\to\mathcal{T}$ which outputs a tensor $\hat{\mathbf{T}}$ in the alphabet $\mathcal{T}$ from the compressed index $m$, namely $\hat{\mathbf{T}} = \psi_n(m)$.
\end{itemize}
\end{definition}
\begin{definition}
A compression threshold $C$ is almost-losslessly achievable if there exists a sequence of tensor compression-reconstruction schemes $(\phi_n,\psi_n)$ satisfying
\begin{displaymath}
\limsup_{n\to\infty} \frac{1}{n}\log|\mathcal{M}|\leq C, \;\lim_{n\to\infty}\mathrm{Pr}\{\hat{\mathbf{T}}\neq \mathbf{T}\}\to 0,
\end{displaymath}
where $|\mathcal{M}|$ is the cardinality of $\mathcal{M}$. We define $C^{\star}_t$ to be the infimum of all almost-losslessly achievable thresholds $C$.
\end{definition}
Conventionally we would be interested in the \textit{compression rate} \cite{cover1999elements} defined as $\frac{1}{\text{input size}}\log|\mathcal{M}|$. In our setting the input size is given by $n^N$ and the compression alphabet size is upper bounded by $(\prod_i|\mathcal{X}_i|)^{nR}$. Therefore the compression rate is zero when $N\geq 2$. However, in practice we are interested in the amount of information that we need to store rather than the rate alone. Therefore the quantity compression threshold is appropriate in this case. Finally to characterize the minimum compression threshold we need to use the entropy. For a distribution $P$ on a finite alphabet $\mathcal{X}$ the entropy $H(P)$ is defined as $H(P) = \sum_{x\in\mathcal{X}} -P(x)\log P(x)$.

\begin{remark}

In practice, performing exact CP decomposition of $\mathbf{T}$, is a challenging problem due to non-linearity. Therefore it is difficult to choose the exact alphabets $\mathcal{X}_i$ for modeling.  For some data sources such as videos or images, as pixels take values between 0 and 255, we can assume that entries of the tensors $\mathbf{T}$ take values on a finite alphabet. Hence the distribution of $\mathbf{T}$ is a discrete one. It is then natural to model $\mathcal{X}_i$ to be discrete albeit not easy to select.

If we assume that $\mathbf{T}$ comes from a continuous distribution then it is also often assumed that entries of factors take values on $\mathbb{R}$. Hence our assumption that $\mathcal{X}_i$ is finite, can be seen as a quantization argument. This naturally introduces distortion into our formulation. Certainly, we want to obtain a compression-distortion trade off. Our study however indicates that obtaining this goal might be formidable.

Despite the above shortcomings our model is practically useful in the following sense. Suppose that a universal (almost lossless) compression algorithm is designed for tensors. Our model can be used as an additional performance benchmark. For example tensors of size $n^N$ can be artificially generated according to our model and provided as the input for the algorithm. If the compression threshold obtained by the algorithm is far from the minimum threshold then there is still room for improvement. It should be noted that algorithms should not be designed specifically for our model due to its simplified assumptions.

\end{remark}

\begin{remark}
Although it is not the focus of our study, let us consider the case that for all $i\in [1:N]$ except one $I_i$ are constant while the last dimension grows. This models the case where one single dimension is very large while the others are very small, e.g., in a long video recording session. Tensor decomposition reduces the number of parameters that need to be stored from $\prod_{i=1}^NI_i$ to $\sum_i I_i$, which is not very satisfying. One should not stop there and instead look for new methods that compress the tensor further.  We can adapt our theory to this problem but it would be more complicated. Even when $I_i$ are large for all $i\in [1:N]$, it is important to compress the tensor further since the data might be distributed to a large number of users.
\end{remark}

\section{Rank-one compression}\label{sec_3}
We begin our study by analyzing the simplest case of rank-one tensor compression. In this we explain our main information theoretic idea for unfamiliar audiences. This section hence serves as a warming up to more complex scenarios in latter sections.\\
Recall that our random tensor $\mathbf{T}$ can be written as
\begin{equation}
T_{i_1,\dots,i_N} = a_{1i_1}\cdots a_{Ni_N}.\label{rank_one_assump}
\end{equation}
The following result characterizes the fundamental compression threshold for the single-component, i.e., rank-one, tensor scenario.
\begin{theorem}\label{thm_1}
When $\mathbf{T}$ is a single-component tensor, then the minimum almost-lossless compression threshold is given by
\begin{displaymath}
C^{\star}_t = \sum_{i=1}^N H(P_{\mathcal{X}_i}).
\end{displaymath}
\end{theorem}
\begin{proof}
First we show that for any $\eta>0$, $\sum_{i=1}^N H(P_{\mathcal{X}_i})+\eta$ is an achievable compression threshold, i.e., $C^{\star}_t \leq  \sum_{i=1}^N H(P_{\mathcal{X}_i})$ holds.\\
In information theory to show that a sequence of mappings $(\phi_n,\psi_n)$ exists, one usually uses a typicality argument. Namely we only need to consider a small set of tensors $\mathbf{T}$ which takes most of the probability. Our setting is a non-iid setting, as there are correlations between elements of $\mathbf{T}$. We therefore need some more work than conventional approaches in \cite{hanspectrum,cover1999elements}. Given a $\gamma>0$, for each $i\in [1:N]$ we define a typical set 
\begin{displaymath}
\mathcal{A}_{i,\gamma}^n = \{\mathbf{a}_i\mid |-\log P_{\mathcal{X}_i}^n(\mathbf{a}_i) - nH(P_{\mathcal{X}_i})|<n\gamma\}.
\end{displaymath}
For each $i \in [1:N]$, to store an element inside $\mathcal{A}_{i,\gamma}^n$ one needs at most $n(H(P_{\mathcal{X}_i}) + \gamma)$ nats. By \cite[Theorem 3.1.2]{cover1999elements} we also have 
\begin{displaymath}
P_{\mathcal{X}_i}^n(\mathcal{A}_{i,\gamma}^n) \geq (1-\gamma),\;\text{for all sufficiently large}\; n.
\end{displaymath}
We define the Cartesian product set $\mathcal{S}_{\gamma,n} = \bigtimes_{i=1}^N \mathcal{A}_{i,\gamma}^n$. Then the typical set used for compression in our setting is defined as
\begin{displaymath}
\mathcal{T}_{\gamma,n} = \{\mathbf{a}_1\circ\cdots\circ\mathbf{a}_N\mid (\mathbf{a}_1,\dots,\mathbf{a}_N)\in \mathcal{S}_{\gamma,n}\}.
\end{displaymath}
 We observe  
\begin{align}
\mathrm{Pr}\{\mathbf{T}\in \mathcal{T}_{\gamma,n}\} &\geq \sum_{(\mathbf{a}_1,\dots,\mathbf{a}_N)\in \mathcal{S}_{\gamma,n}} \prod_{i=1}^N P_{\mathcal{X}_i}^n(\mathbf{a}_i) = \prod_{i=1}^N P_{\mathcal{X}_i}^n(\mathcal{A}_{i,\gamma}^n)\nonumber\\
&\geq (1-\gamma)^N.
\end{align}
The first inequality follows since there might exist tuples $(\mathbf{a}_1,\dots,\mathbf{a}_N)\in \mathcal{S}_{\gamma,n}$ and $(\mathbf{a}_1^{\prime},\dots,\mathbf{a}_N^{\prime})\notin \mathcal{S}_{\gamma,n}$ such that 
\begin{displaymath}
\mathbf{a}_1\circ\cdots\circ\mathbf{a}_N = \mathbf{a}_1^{\prime}\circ\cdots\circ\mathbf{a}_N^{\prime}.
\end{displaymath} 
Our compression scheme is given as follows. $\phi_n$ maps each tensor in $\mathcal{T}_{\gamma,n}$ to a unique index. The reconstruction mapping $\psi_n$ maps the stored index to the corresponding tensor. For tensors $\mathbf{T}$ that do not belong to $\mathcal{T}_{\gamma,n}$, $\phi_n$ maps these to a fixed index and $\psi_n$ outputs a fixed tensor. This requires at most $n(\sum_{i=1}^N H(P_{\mathcal{X}_i})+N\gamma)$ nats.
Therefore for any given $\eta>0$, there exists a $n_0(\eta)$ such that for all sufficiently large $n\geq n_0(\eta)$ our tensor compression-reconstruction scheme satisfies
\begin{align}
\mathrm{Pr}\{\hat{\mathbf{T}}\neq \mathbf{T}\}\leq \eta, \quad \frac{1}{n}\log |\mathcal{M}|\leq \sum_{i=1}^N H(P_{\mathcal{X}_i})+ \eta.
\end{align}
In conclusion we have $C^{\star}_t \leq  \sum_{i=1}^N H(P_{\mathcal{X}_i})$.\\
Now we show the reverse direction, i.e., $C^{\star}_t \geq  \sum_{i=1}^N H(P_{\mathcal{X}_i})$. For a given achievable compression threshold $C$ let $\{(\phi_n,\psi_n)\}_{n=1}^{\infty}$ be a given sequence of tensor compression-reconstruction schemes satisfying
\begin{displaymath}
\limsup_{n\to\infty}\frac{1}{n}\log|\mathcal{M}|\leq C,\; \lim_{n\to\infty}\mathrm{Pr}\{\psi_n(\phi_n(\mathbf{T}))\neq \mathbf{T}\}\to 0.
\end{displaymath}
Then for a given $\epsilon>0$ there exists a $n_{\epsilon}$ such that for all $n\geq n_{\epsilon}$ we have
\begin{displaymath}
\log|\mathcal{M}|\leq n(C+\epsilon),\;\text{and}\;\mathrm{Pr}\{\psi_n(\phi_n(\mathbf{T}))\neq \mathbf{T}\} \leq \epsilon.
\end{displaymath}
For a given $n$ we define the following correct decodable set of tensors
\begin{displaymath}
\mathcal{D}_n = \{ \mathbf{T}\mid \psi_n(\phi_n(\mathbf{T})) = \mathbf{T}\}.
\end{displaymath}
We define the expanding set of factors
\begin{displaymath}
\mathcal{S}_n = \{(\mathbf{a}_1,\dots,\mathbf{a}_N)\mid \mathbf{a}_1\circ\cdots\circ\mathbf{a}_N\in \mathcal{D}_n\}.
\end{displaymath}
We have
\begin{align}
 \mathrm{Pr}\{(\mathbf{X}_1,\dots,\mathbf{X}_N)\in \mathcal{S}_n\} = \mathrm{Pr}\{\mathbf{T}\in\mathcal{D}_n\}. \label{condensed_prob}
 \end{align}
  We observe that for a given tensor $\mathbf{T}$, once $a_{i1}$ are given for all but one $i\in [1:N]$ we can deduce the remaining elements in vectors $\{\mathbf{a}_j\}_{j=1}^N$. Therefore we have $|\mathcal{S}_n|\leq \min_i \prod_{j\neq i}|\mathcal{X}_j||\mathcal{D}_n|$. Additionally since $\phi_n$ can only take $|\mathcal{M}|$ values so does the composite mapping $\psi_n(\phi_n(\cdot))$, which implies that $|\mathcal{D}_n|\leq |\mathcal{M}|$. In conclusion we have $|\mathcal{S}_n|\leq \min_i \prod_{j\neq i}|\mathcal{X}_j| |\mathcal{M}|$. \\
Since the distribution of $\mathbf{T}$ is no longer a product of identical components in our case, standard arguments using Fano's inequality as in \cite{cover1999elements} are no longer applicable. In order to show the converse we define the following atypical set, parameterized by $n$ and $\eta$,
\begin{align}
\mathcal{T}_{n,\eta} = \{(\mathbf{a}_1,\dots,\mathbf{a}_N) \mid -\sum_{i=1}^N\log P_{\mathcal{X}_i}^n(\mathbf{a}_i)\geq \log|\mathcal{M}|+ \eta\}.\label{typical_converse}
\end{align}
For all sufficiently large $n\geq n_{\epsilon}$, using the information-spectrum arguments \cite{hanspectrum} we then have
\begin{align}
\mathrm{Pr}\{(\mathbf{X}_1,\dots,\mathbf{X}_N) \in \mathcal{T}_{n,\eta}\} &= \mathrm{Pr}\{(\mathbf{X}_1,\dots,\mathbf{X}_N) \in \mathcal{S}_n\cap \mathcal{T}_{n,\eta}\} + \mathrm{Pr}\{(\mathbf{X}_1,\dots,\mathbf{X}_N) \in \mathcal{S}_n^c\cap \mathcal{T}_{n,\eta}\}\nonumber\\
&\leq \mathrm{Pr}\{(\mathbf{X}_1,\dots,\mathbf{X}_N) \in \mathcal{S}_n^c\} + \mathrm{Pr}\{(\mathbf{X}_1,\dots,\mathbf{X}_N) \in \mathcal{S}_n\cap \mathcal{T}_{n,\eta}\}\nonumber\\
& \stackrel{\eqref{condensed_prob}}{=} \mathrm{Pr}\{\psi_n(\phi_n(\mathbf{T}))\neq \mathbf{T}\} + \sum_{(\mathbf{a}_1,\dots,\mathbf{a}_N)\in \mathcal{S}_n\cap \mathcal{T}_{n,\eta}}\prod_{i=1}^N P_{\mathcal{X}_i}^n(\mathbf{a}_i)\nonumber\\
&\stackrel{\eqref{typical_converse}}{\leq} \epsilon + |\mathcal{S}_n\cap \mathcal{T}_{n,\eta}|e^{-\eta}/|\mathcal{M}|\leq \epsilon + |\mathcal{S}_n|e^{-\eta}/|\mathcal{M}|\nonumber\\
&\leq \epsilon + e^{-\eta}\min_i\prod_{j\neq i}|\mathcal{X}_j|.
\end{align}
We take $\eta = n\gamma$. If for all $n_0\geq n_{\epsilon}$ there exists a $n>n_0$ such that $\log |\mathcal{M}| + \eta< n(\sum_i H(P_{\mathcal{X}_i})- N\gamma)$, then we have 
\begin{displaymath}
\limsup_{n\to\infty} \mathrm{Pr}\{(\mathbf{X}_1,\dots,\mathbf{X}_N) \in \mathcal{T}_{n,\eta}\} = 1,
\end{displaymath}
due to the weak law of large numbers,
which violates that last inequality. We must have $\log |\mathcal{M}| + n\gamma\geq n (\sum_{i}H(P_{\mathcal{X}_i}) - N\gamma)$, which in turn implies that $n(C+\epsilon+\gamma)\geq n (\sum_{i}H(P_{\mathcal{X}_i}) - N\gamma)$, for all $n\geq n_{\epsilon}$. Since $\gamma$ and $\epsilon$ are arbitrary we have $C^{\star}_t\geq \sum_{i}H(P_{\mathcal{X}_i})$.
\end{proof}
\begin{remark}
In the converse direction, we have carefully controlled the contribution of atypical tuples of factors, those who are in $\mathcal{T}_{n,\eta}$. In the single-component scenario, as we have seen, for each tensor the number of such tuples of factors is always bounded by a constant. This no longer holds in the multi-component scenarios.
\end{remark}

\begin{remark}
Due to our assumption on finiteness of $\{\mathcal{X}_i\}_{i=1}^N$ and the relation \eqref{rank_one_assump}, performing tensor decomposition of the rank-one model is relatively easy. The output of a decomposition rule $\{\hat{\mathbf{a}}_i\}_{i=1}^N$ is however not necessarily the same as the underlying $\{\mathbf{a}_i\}_{i=1}^N$. Straightforward usage of compression schemes designed for $P_{\mathcal{X}_i}$ on $\hat{\mathbf{a}}_i$ is not recommended due to a distribution mismatch. In other words, a two-step algorithm involving a tensor decomposition in the first step and compression of factors in the second step might not be optimal. Our scheme indicates that to ensure optimality at least some additional constraints, such as typicality in our model, need to be imposed on top of the tensor decomposition.
\end{remark}

\section{Examples of 2-component compression}\label{sec_4}
In this section we present some examples involving random tensors admitting two-component decompositions.
In these examples we show that there exist tensors for which the number of tuples of factors increases exponentially with $n$. This implies that a straightforward application of previous arguments in Theorem \ref{thm_1} is no longer possible. However the probability of the set of these tensors are negligible. Hence the converse arguments can be fixed.

\textit{Example 1}: In the first example we consider a supersymmetric scenario where the order of the tensor is 3, $N=3$, and
\begin{displaymath}
\mathbf{X}_1 = \mathbf{X}_2 = \mathbf{X}_3 = [\mathbf{a}_1,\mathbf{a}_2] \in \mathcal{X}^{n\times 2}.
\end{displaymath}
We further assume that the underlying alphabet $\mathcal{X}$ is $\mathcal{X} = \{-1,1\}$. Assume that $a_{1i}\sim P$ and $a_{2i}\sim Q$ where $P$ and $Q$ are two distributions on $\mathcal{X}$ satisfying $P(a)\neq 0$ and $Q(a)\neq 0$ for all $a\in \mathcal{X}$. To design a compression mapping one only needs to look at the set
\begin{displaymath}
\mathcal{A}_{\gamma}^n = \{\mathbf{X}\mid |-\log \big[P^n(\mathbf{a}_1)Q^n(\mathbf{a}_2)\big] - n[H(P)+H(Q)]|<n\gamma\}.
\end{displaymath}
We similarly form a set
\begin{displaymath}
\mathcal{T}_{n,\gamma} = \{[\mathbf{X};\mathbf{X};\mathbf{X}]\mid\mathbf{X}\in\mathcal{A}_{\gamma}^n\},
\end{displaymath}
and index all of its elements. If $\mathbf{T}$ belongs to $\mathcal{T}_{n,\gamma}$ we store the corresponding index. Otherwise we store a special index for all tensors that are not in $\mathcal{T}_{n,\gamma}$. This requires $n[H(P)+H(Q)+\gamma]$ nats.\\
 In the converse direction given a realization $\mathbf{T}$ we bound the number of pairs of factors $(\mathbf{a}_1,\mathbf{a}_2)$ resulting in $\mathbf{T}$. Each element of the tensor $\mathbf{T}$ is given by
\begin{displaymath}
T_{i_1i_2i_3} = a_{1i_1}a_{1i_2}a_{1i_3} + a_{2i_1}a_{2i_2}a_{3i_3}.
\end{displaymath}
Since the alphabet is $\{-1,1\}$ we have $a^3 = a$ for any $a\in\mathcal{X}$. To recover $(\mathbf{a}_1,\mathbf{a}_2)$ from $\mathbf{T}$ we therefore only need to consider $n+\binom{n}{3}$ expressions of the forms $a_{1i} + a_{2i} = T_{iii} $ where $i \in [1:n]$, and $a_{1i_1}a_{1i_2}a_{1i_3} + a_{2i_1}a_{2i_2}a_{2i_3} = T_{i_1i_2i_3}$ where $(i_1,i_2,i_3)\in [1:n]^3$ are mutually different.
\begin{itemize}
\item Let us consider the first case when $a_{1i} + a_{2i} = 0$ for all $i\in [1:n]$. This implies that $a_{1i} = -a_{2i}$ for all $i\in [1:n]$. For any triple $(i_1,i_2,i_3)$ we then have $a_{1i_1}a_{1i_2}a_{1i_3} = - a_{2i_1}a_{2i_2}a_{2i_3}$ which leads to $a_{1i_1}a_{1i_2}a_{1i_3} + a_{2i_1}a_{2i_2}a_{2i_3} = 0$. Therefore when $\mathbf{T} = \mathbf{0}$, the number of decompositions is controlled by the system of equations $\{a_{1i} + a_{2i} =0\}_{i=1}^n$. The number of pairs of factors that result in this particular tensor is hence $2^{n}$. In more detail each solution factor matrix has the form
\begin{displaymath}
\mathbf{X} = [\mathbf{a},-\mathbf{a}], \; \mathbf{a}\in \mathcal{X}^n,
\end{displaymath}
i.e., a \textit{rank-deficient} matrix. Next, we will calculate the probability of the event $\mathbf{T} = \mathbf{0}$. Then we have
\begin{align*}
\mathrm{Pr}\{\mathbf{T}=\mathbf{0}\} &= \sum_{\mathbf{a}\in\mathcal{X}^n}P^n(\mathbf{a})Q^n(-\mathbf{a})= \big[\sum_{a\in \mathcal{X}}P(a)Q(-a)\big]^n.
\end{align*}
As 
\begin{displaymath}
\sum_{a\in \mathcal{X}}P(x)Q(-x) = P(1)Q(-1) + P(-1)Q(1)<\max\{Q(-1),Q(1)\}<1
\end{displaymath}
we obtain that $\mathrm{Pr}\{\mathbf{T}=\mathbf{0}\}\to 0$ as $n\to\infty$.
\item Consider a tensor $\mathbf{T}$ for which we have, without the loss of generality,
\begin{align*}
 a_{1n} + a_{2n} &= 2,\nonumber\\
 \text{and}\; a_{1i} + a_{2i} &= 0,\; \forall i\in [1:n-1].
 \end{align*}
Then for all $i\in [1:n-1]$  since $a_{1i} = - a_{2i}$ holds we have
\begin{displaymath}
T_{1in} = a_{11}a_{1i}a_{1n} + a_{21}a_{2i}a_{2n} = 2a_{11}a_{1i}.
\end{displaymath}
 This implies that for a given choice of $a_{11}$ and a given $\mathbf{T}$ we can infer the other values of $a_{1i}$ and $a_{2i}$ uniquely. The number of pairs of factors for a given $\mathbf{T}$ in this case is $2$. Having more constraints of the form $a_{1i} + a_{2i}\neq 0$ does not increase the number of decompositions for a similar reason.
\item When $a_{1i} + a_{2i}\neq 0$ for all $i\in [1:n]$ then there is a unique pair of factors $(\mathbf{a}_1,\mathbf{a}_2)$ resulting in $\mathbf{T}$.
\end{itemize}
By excluding the all 0 tensor, $\mathbf{T}= \mathbf{0}$, from the decodable set $\mathcal{D}_n$, since the event has a vanishing probability, we can apply a similar argument as in the converse proof of Theorem \ref{thm_1}, for example the atypical set $\mathcal{T}_{n,\eta}$ can be defined accordingly as
\begin{displaymath}
\mathcal{T}_{n,\eta} = \{(\mathbf{a}_1,\mathbf{a}_2) \mid -(\log P^n(\mathbf{a}_1)+\log Q^n(\mathbf{a}_2))\geq \log|\mathcal{M}|+ \eta\},
\end{displaymath}
to conclude that in this case the minimum compression threshold is given by $C^{\star}_t = H(P) + H(Q)$.

\textit{Example 2}: Let us consider the compression of an order-2 tensor admitting a two single component decomposition
\begin{displaymath}
\mathbf{T} = \mathbf{X}_1\mathbf{X}_2^T = \mathbf{x}\mathbf{y}^T + \mathbf{u}\mathbf{v}^T,
\end{displaymath}
where $(\cdot)^T$ is the transpose operation and for notation brevity we have abbreviated $\mathbf{X}_1 = [\mathbf{x},\mathbf{u}]$ as well as $\mathbf{X}_2 = [\mathbf{y},\mathbf{v}]$. We assume that all random variables take values in the set $\{-1,1\}$ as well as $\mathbf{x}\sim P_{X}^{n}$, $\mathbf{y}\sim P_{Y}^n$, $\mathbf{u}\sim P_U^n$ and $\mathbf{v}\sim P_V^n$. This implies that elements of $\mathbf{T}$ take values in the set $\{-2,0,2\}$. To derive the informtion-theoretic converse for compression of this model we are similar interested in the number of tuples of factors $(\mathbf{x},\mathbf{y},\mathbf{u},\mathbf{v})$ resulting in a given tensor $\mathbf{T}$. 

Let $\mathbf{t}\in \{-2,0,2\}^{n^2\times 1}$ be the vectorized version of the transpose of $\mathbf{T}$. Without the loss of generality we assume that $\mathbf{t}$ has the following form
\begin{displaymath}
\mathbf{t} = (\underbrace{0,\dots,0}_{k},\underbrace{2,\dots,2}_{l},\underbrace{-2,\dots,-2}_{n-k-l},T_{21},T_{22},\dots)^T.
\end{displaymath}
A complete analysis of this example consists of the following major cases:
\begin{itemize}
\item $k\in [1:n-1]$, i.e., the first row contains at least one 0 and one non-zero,
\item $k=0$, i.e., none of the elements in the first row is zero,
\item $k=n$, i.e., the first row is the zero vector.
\end{itemize}
Presenting the entire details is rather complex and unnecessary. We consider two representative scenarios occurring when $k=n$. Let $m$ the number of rows indexed by $\{j_1,\dots,j_m\}$ such that $T_{ji}\neq 0$ for all $i\in [1:n]$ and $j\in \{j_1,\dots,j_m\}$. 
\begin{itemize}
\item When $\mathbf{T} = \mathbf{0}$, for a given $x_1,u_1 \in \{-1,1\}$ we have
\begin{align}
x_1y_i + u_1v_i = 0&\implies v_i = -u_1x_1y_i,\quad \forall i\in [1:n],\nonumber\\
x_iy_1 + u_iv_1 = 0&\implies u_i = x_1u_1x_i,\quad \forall i\in [2:n].
\end{align}
The factor matrices have the following form $\mathbf{X}_1 = [\mathbf{x}, a\mathbf{x}]$ and $\mathbf{X}_2 = [\mathbf{y},-a\mathbf{y}]$ where $a\in \{-1,1\}$, $\mathbf{x}\in \{-1,1\}^n$, $\mathbf{y}\in \{-1,1\}^n$. They are \textit{rank-deficient}. The number of tuples of factors for the all 0 tensor is $2\times 2^{n}\times 2^{n} = 2^{2n+1}$. The probability of this event is given by
\begin{align*}
\mathrm{Pr}\{\mathbf{T} = \mathbf{0}\} &= \sum_{\mathbf{x},\mathbf{y},a}P_{X}^n(\mathbf{x})P_Y^n(\mathbf{y})P_{U}^{n}(a\mathbf{x})P_V^n(-a\mathbf{y})\nonumber\\
&= \sum_{a}\sum_{\mathbf{x}}P_X^n(\mathbf{x})P_{U}^n(a\mathbf{x})\sum_{\mathbf{y}}P_{Y}^n(\mathbf{y})P_V^n(-a\mathbf{y})\nonumber\\
&=\sum_{a}\big[\sum_x P_{X}(x)P_{U}(ax)\big]^n\big[\sum_y P_{Y}(y)P_V(-ay)\big]^n\nonumber\\
& = \big[\sum_xP_{X}(x)P_{U}(x)\big]^n\big[\sum_yP_{Y}(y)P_V(-y)\big]^n\nonumber\\
&\quad + \big[\sum_xP_{X}(x)P_{U}(-x)\big]^n\big[\sum_yP_{Y}(y)P_V(y)\big]^n.
\end{align*}
We have
\begin{displaymath}
\mathrm{Pr}\{\mathbf{T}= \mathbf{0}\}\to 0\; \text{as}\; n\to \infty.
\end{displaymath}
\item When $m\geq 1$ holds, then given $(y_1,v_1)$, $(y_2,\dots,y_n)$ and $(v_2,\dots,v_n)$ are determined through
\begin{displaymath}
x_{j_1} = \text{sign}(T_{j_1i})y_i,\; u_{j_1} = \text{sign}(T_{j_1i})v_i,\;\forall i\in [1:n].
\end{displaymath}
For $j\notin \{j_1,\dots,j_m\}$ we can select $x_j$ freely. Therefore for each of these tensors we can find $8\times 2^{n-1-m} = 2^{n-m+2}$ tuples of factors $(\mathbf{x},\mathbf{y}, \mathbf{u},\mathbf{v})$ resulting in it. Furthermore due to the structure we also have $(T_{j_i1},\dots,T_{j_in}) = \pm (T_{j_11},\dots,T_{j_1n})$ for all $i\in [2:m]$.
A representative tuple of factors has the following form
\begin{align*}
\mathbf{x} &= (x_1,\beta_1x_1,\dots,\beta_mx_1,x_{m+2},\dots,x_n)^T,\nonumber\\
\mathbf{u} &= (-u_1,\beta_1u_1,\dots,\beta_mu_1,-x_1u_1x_{m+2},\dots,-x_1u_1x_n)^T,\nonumber\\
\mathbf{y} &= (\alpha_1x_1,\dots,\alpha_nx_1)^T,\nonumber\\
\mathbf{v} &= (\alpha_1u_1,\dots,\alpha_nu_1)^T,
\end{align*}
where $\beta_j\in \{-1,1\}$ for all $j\in [1:m]$ and $\alpha_i \in \{-1,1\}$ for all $i\in [1:n]$. We observe that the second factor matrix $\mathbf{X}_2 = [\mathbf{y},\mathbf{v}]$ is also \textit{rank-deficient}. We calculate the contribution of this event in the following. 
The probability of the given tuple of factors is
\begin{align*}
&P_X(x_1)P_U(-u_1)P_X^m(x_1\bm{\beta})P_{U}^m(u_1\bm{\beta})\nonumber\\
&\times P_X^{n-m-1}(\mathbf{x}_{m+2}^n)P_U^{n-m-1}(-x_1u_1\mathbf{x}_{m+2}^n)P_Y^n(x_1\bm{\alpha})P_V^n(u_1\bm{\alpha}),
\end{align*}
where $\mathbf{x}_{m+2}^n=(x_{m+2},\dots,x_n)$. Summing over $\bm{\alpha}$ we obtain
\begin{displaymath}
\sum_{\bm{\alpha}}P_Y^n(x_1\bm{\alpha})P_V^n(u_1\bm{\alpha}) = [\sum_{\alpha}P_Y(x_1\alpha)P_V(u_1\alpha)]^n.
\end{displaymath}
Summing over $\bm{\beta}$ we obtain
\begin{displaymath}
\sum_{\bm{\beta}}P_X^m(x_1\bm{\beta})P_{U}^m(u_1\bm{\beta}) = [\sum_{\beta}P_X(x_1\beta)P_U(u_1\beta)]^m.
\end{displaymath}
Summing over $\mathbf{x}_{m+2}^n$ we obtain
\begin{displaymath}
\sum_{\mathbf{x}_{m+2}^n}P_X^{n-m-1}(\mathbf{x}_{m+2}^n)P_U^{n-m-1}(-x_1u_1\mathbf{x}_{m+2}^n) = \big[\sum_xP_X(x)P_U(-x_1u_1x)\big]^{n-m-1}.
\end{displaymath}
Finally by summing over the all possible choices of $\{j_1,\dots,j_m\}$, the number of rows $m$, $x_1$ and $u_1$, the total probability of this event is hence
\begin{align*}
&\sum_{x_1,u_1}P_X(x_1)P_U(-u_1)[\sum_{\alpha}P_Y(x_1\alpha)P_V(u_1\alpha)]^n\nonumber\\
&\times \sum_{m=1}^{n-1}\binom{n-1}{m}[\sum_{\beta}P_X(x_1\beta)P_U(u_1\beta)]^m[\sum_xP_X(x)P_U(-x_1u_1x)]^{n-m-1}\nonumber\\
& = [\sum_xP_X(x)P_U(-x)][\sum_{\alpha}P_Y(\alpha)P_V(\alpha)]^n\nonumber\\
&\times \sum_{m=1}^{n-1}\binom{n-1}{m}[\sum_{\beta}P_X(\beta)P_U(\beta)]^m[\sum_xP_X(x)P_U(-x)]^{n-m-1}\nonumber\\
&+[\sum_xP_X(x)P_U(x)][\sum_{\alpha}P_Y(\alpha)P_V(-\alpha)]^n\nonumber\\
&\times \sum_{m=1}^{n-1}\binom{n-1}{m}[\sum_{\beta}P_X(\beta)P_U(-\beta)]^m[\sum_xP_X(x)P_U(x)]^{n-m-1}\nonumber\\
&\leq [\sum_xP_X(x)P_U(-x)][\sum_{\alpha}P_Y(\alpha)P_V(\alpha)]^n \nonumber\\
&+ [\sum_xP_X(x)P_U(x)][\sum_{\alpha}P_Y(\alpha)P_V(-\alpha)]^n\to 0.
\end{align*}
The last inequality is valid since the following reduction holds
\begin{align*}
\sum_{m=1}^{n-1}&\binom{n-1}{m}[\sum_{\beta}P_X(\beta)P_U(\beta)]^m[\sum_xP_X(x)P_U(-x)]^{n-m-1}\nonumber\\
&  = \big(\sum_{\beta}P_X(\beta)P_U(\beta) + \sum_xP_X(x)P_U(-x)\big)^{n-1} - \big(\sum_xP_X(x)P_U(-x)\big)^{n-1}\nonumber\\
&= 1-\big(\sum_xP_X(x)P_U(-x)\big)^{n-1}.
\end{align*}

\end{itemize}
Using similar lines of arguments we can show that in this case the minimum compression threshold is $C^{\star}_t = H(P_X) + H(P_Y) + H(P_U) + H(P_V)$.

\section{Compression of multi-component tensor}\label{sec_5}
Recall that a multi-component tensor has the following form
\begin{displaymath}
\mathbf{T} = [\mathbf{X}_1;\cdots;\mathbf{X}_N] = \sum_{r=1}^R \mathbf{a}_{r1}\circ\cdots\circ\mathbf{a}_{rN},
\end{displaymath}
where $R\geq 2$ holds. 
In the last section we have seen that for some tensors such as the all 0 tensor, the number of tuples of factors grows exponentially with $n$. We also observe that some factor matrices in these cases are rank-deficient. Therefore, a workaround idea for the general scenario would be restricting our attention to the set of tuples of full rank factor matrices, i.e., all factor matrices in a given tuple are full rank. This does not immediately guarantee that there would not exist a set of tensors $\mathbf{T}$ admitting a growing number of factorizations which perhaps has non-vanishing probability. Fortunately, our analysis in the following shows that the case does not occur. In the following we first show that a random factor matrix is full rank with high probability.
\begin{lemma}\label{lemma_1}
For each $i\in [1:N]$ we have 
\begin{displaymath}
\mathrm{Pr}\{\mathrm{rank}(\mathbf{X}_i)=R\}\to 1,\;\text{as}\; n\to\infty.
\end{displaymath}
\end{lemma} 
 When $P_{\mathcal{X}_{i},r} = P_i$ for all $r\in [1:R]$ the result can be deduced from the fact that the probability of a square random matrix with iid elements being singular is vanishing \cite{komlos1968determinant}. For our setting we use arguments in \cite{tao2006random}. 
\begin{proof} 
For a given $i\in [1:N]$, $\mathbf{X}_i$ is rank-deficient, i.e., $\mathrm{rank}(\mathbf{X}_i)< R$, implies that one of the following events happens
\begin{displaymath}
\mathcal{A}_{i0} = \{\mathbf{X}_{i1} = \mathbf{0}\},\quad \mathcal{A}_{ir} = \{\mathbf{X}_{i(r+1)}\in \mathrm{span}(\mathbf{X}_{i1},\dots,\mathbf{X}_{ir})\},\; r=1, \dots, R-1.
\end{displaymath}
First of all we have
\begin{displaymath}
\mathrm{Pr}\{\mathcal{A}_{i0}\} = P_{\mathcal{X}_{i,1}}(0)^n <\rho_i^n,
\end{displaymath}
where $\rho_i = \max_{x,k}P_{{\mathcal{X}}_i,k}(x)<1$ by our assumption. The inequality also holds when $0\notin \mathcal{X}_i$ occurs.
For each $r=1,\dots,R-1,$ by summing over all possible vector subspaces $V$, which is finite, we have
\begin{displaymath}
\mathrm{Pr}\{\mathcal{A}_{ir}\cap\mathcal{A}_{i0}^c\} = \sum_{V}\mathrm{Pr}\{\mathrm{span}(\mathbf{X}_{i1},\dots,\mathbf{X}_{ir}) = V,\mathbf{X}_{i1} \neq \mathbf{0}\}\mathrm{Pr}\{\mathbf{X}_{i(r+1)}\in V\}.
\end{displaymath}
Given a vector subspace $V$ with dimension $t=\text{dim(V)}\in [1:r]$, we can determine each vector $\mathbf{v}\in V$ completely based on $t$ coordinates, for instance we can calculate $(v_{t+1},\dots,v_n)$ based on $(v_1,\dots,v_t)$. Without the loss of generality we denote this relation by $(v_{t+1},\dots,v_n) = f_V(v_1,\dots,v_t)$. Therefore in this case
\begin{align}
\mathrm{Pr}\{\mathbf{X}_{i(r+1)}\in V\} \leq &\sum_{(x_1,\dots,x_t)\in\mathcal{X}_i^t}\mathrm{Pr}\{(X_{i(r+1),1},\dots,X_{i(r+1),t})=(x_1,\dots,x_t)\}\nonumber\\
&\times \mathrm{Pr}\{(X_{i(r+1),t+1},\dots,X_{i(r+1),n})=f_V(x_1,\dots,x_t)\}\nonumber\\
&\leq \sum_{(x_1,\dots,x_t)\in\mathcal{X}_i^t}\mathrm{Pr}\{(X_{i(r+1),1},\dots,X_{i(r+1),t})=(x_1,\dots,x_t)\}\times \rho_i^{n-t}\nonumber\\
& = \rho_i^{n-t}\leq \rho_i^{n-r}
\end{align}
as $\rho_i<1$ holds.
In summary we obtain
\begin{align}
\mathrm{Pr}\{\mathrm{rank}(\mathbf{X}_i)\leq R\}&\leq \mathrm{Pr}\{\mathcal{A}_{i0}\} + \sum_{r=1}^{R-1}\mathrm{Pr}\{\mathcal{A}_{ir}\cap \mathcal{A}_{i0}^c\}\nonumber\\
&\leq \sum_{r=0}^{R-1}\rho_i^{n-r} = \zeta_i\to 0, \;\text{as}\; n\to\infty.
\end{align}
\end{proof}
We need a link between the full rank property of each factor matrix and the number of tuples of full rank factor matrices resulting in a given tensor. In case $N\geq 3$ this link is established through the Kruskal's rank (or k-rank) introduced in \cite{kruskal1977three}.
\begin{definition}
The $k$-rank of a matrix $\mathbf{A}$, $k_{\mathbf{A}}$, is the largest value of $t$ such that every subset of $t$ columns of $\mathbf{A}$ is linearly independent.
\end{definition}
We always have $k_{\mathbf{A}}\leq \mathrm{rank}(\mathbf{A})$. For a full-column rank matrix $\mathbf{A}$, we have $k_{\mathbf{A}} = \mathrm{rank}(\mathbf{A})$. Lemma \ref{lemma_1} implies the following result.
\begin{corollary}\label{coroll_1}
 For each $i\in [1:N]$ we have
\begin{displaymath}
\mathrm{Pr}\{k_{\mathbf{X}_i}=R\} \geq 1-\zeta_i\to 1,\;\text{as}\; n\to \infty.
\end{displaymath}
\end{corollary}
A sufficient condition based on $k$-rank that enables us to bound the number of tuples of full rank factor matrices yielding the same tensor is given below.
\begin{theorem}[{\cite[Theorem 3]{sidiropoulos2000uniqueness}}]\label{thm_r}
Assume that the order of tensor $N$ satisfies $N\geq 3$. Given deterministic matrices $(\mathbf{X}_i)_{i=1}^N$ of size $I_i\times R$ satisfying $\mathbf{Z} = [\mathbf{X}_1;\cdots;\mathbf{X}_N]$ such that $R$ is the rank of tensor $\mathbf{Z}$, i.e., the minimum number of rank-one decompositions of $\mathbf{Z}$, if 
\begin{align}
 \sum_{i=1}^N k_{\mathbf{X}_i}\geq 2R+(N-1),
\end{align}
Then the matrices $\mathbf{X}_i$ are essentially uniquely determined. 
\end{theorem}
The essential uniqueness means that if $(\mathbf{X}_i^{\prime})_{i=1}^N$ is another tuple of factor matrices satisfying $\mathbf{Z} = [\mathbf{X}_1^{\prime};\cdots;\mathbf{X}_N^{\prime}]$ then for all $i\in [1:N]$, $\mathbf{X}_i^{\prime} = \mathbf{X}_i\mathbf{P}\bm{\Lambda}_i$ holds where $\mathbf{P} \in \{0,1\}^{R\times R}$ is a unique permutation matrix and $\bm{\Lambda}_i$ are unique diagonal matrices of size $R$ satisfying $\prod_{i=1}^N\bm{\Lambda}_i = \mathbf{I}_{R\times R}$. 

We then have the following upper bound on the number of full rank factors yielding the same tensor.
\begin{lemma}\label{lemm_2}
For given $N\geq 2$, $R\geq 2$, $n\geq R$, and finite alphabets $(\mathcal{X}_i)_{i=1}^N$, there exists a number $\Gamma_N$ that does not depend on $n$ such that the number of tuples of full rank matrices $(\mathbf{X}_i)_{i=1}^N$, $\mathbf{X}_i\in\mathcal{X}_i^{n\times R}$ for all $i\in [1:N]$, satisfying $\mathbf{T} = [\mathbf{X}_1;\dots;\mathbf{X}_N]$ is upper bounded by $\Gamma_N$.
\end{lemma}
\begin{proof}
Assume that $N\geq 3$ holds. Note that when $\mathbf{X}_i$ are full rank $R$ for all $i\in [1:N]$, then rank of $\mathbf{T} = [\mathbf{X}_1;\cdots;\mathbf{X}_N]$ is $R$. Assume otherwise that $\mathbf{T} = [\mathbf{X}_1^{\prime};\cdots;\mathbf{X}_N^{\prime}]$ where $\mathbf{X}_i^{\prime}$ is of size $I_i\times (R-1)$ for all $i\in [1:N]$. By unfolding the tensor $\mathbf{T}$ according to the first dimension we obtain
\begin{align}
(\mathbf{X}_N\odot\cdots\odot \mathbf{X}_2)\mathbf{X}_1^T &= (\mathbf{X}_N^{\prime}\odot\cdots\odot \mathbf{X}_2^{\prime}){\mathbf{X}_1^{\prime}}^T,\label{unfolding}
\end{align}
where $\odot$ is the Khatri-Rao product. We need the following property of the Khatri-Rao product. Since $\mathbf{X}_3$ and $\mathbf{X}_2$ are full rank, $k_{\mathbf{X}_3} = k_{\mathbf{X}_2} = R$, we have $\mathrm{rank}(\mathbf{X}_3\odot \mathbf{X}_2)\geq k_{\mathbf{X}_3\odot \mathbf{X}_2}\geq \min\{k_{\mathbf{X}_3} + k_{\mathbf{X}_2} -1, R\}\geq R$ by \cite[Lemma 3.3]{stegeman2007kruskal}. By applying this inequality consecutively we obtain
\begin{displaymath}
k_{\mathbf{X}_i\odot \mathbf{X}_{i-1}\odot\cdot\odot \mathbf{X}_2}\geq \min\{k_{\mathbf{X}_i} + k_{\mathbf{X}_{i-1}\odot\cdot\odot\mathbf{X}_2} -1, R\}\geq R,
\end{displaymath}
since $k_{\mathbf{X}_{i-1}\odot\cdot\odot\mathbf{X}_2}\geq R$ for all $i=4,\dots,N$. Therefore we have
\begin{displaymath}
\mathrm{rank}(\mathbf{X}_N\odot\cdots\odot\mathbf{X}_2) \geq k_{\mathbf{X}_N\odot\cdots\odot\mathbf{X}_2}\geq R.
\end{displaymath}
This gives the contradiction as the the right-hand side of \eqref{unfolding} has rank at most $R-1$ while by the Sylvester’s rank inequality
\begin{displaymath}
\mathrm{rank}((\mathbf{X}_N\odot\cdots\odot\mathbf{X}_2)\mathbf{X}_1^T)\geq \mathrm{rank}(\mathbf{X}_N\odot\cdots\odot\mathbf{X}_2) + \mathrm{rank}(\mathbf{X}_1) - R \geq R
\end{displaymath}
the rank of the left-hand side is $R$. 

Since the rank of $\mathbf{T}$ is $R$, the condition of Theorem \ref{thm_r} is satisfied as $R\geq 1 + 1/(N-2)$ holds. This implies that if $(\mathbf{X}_i^{\prime})_{i=1}^N$, $\mathbf{X}_i^{\prime}\in\mathcal{X}_i^{n\times R}$, $\forall i\in [1:N]$, is another tuple of full rank factor matrices satisfying $\mathbf{T} = [\mathbf{X}_1^{\prime};\cdots;\mathbf{X}_N^{\prime}]$ we then have $\mathbf{X}_i^{\prime} = \mathbf{X}_i\mathbf{P}\bm{\Lambda}_i$ for all $i\in [1:N]$.
Since our alphabets are discrete, for each $i\in [1:N]$ the number of such $\bm{\Lambda}_i$ is finite and does not depend on $n$. We denote the upper bound on the number of different matrices $\mathbf{P}\bm{\Lambda}_i$ by $\Gamma_N$. Hence the conclusion holds in this case.

When $N=2$, i.e., $\mathbf{T}$ is a square matrix, the essential uniqueness does not hold in general. We recall that when $N=2$ a tensor $\mathbf{T}$ can be written as
\begin{displaymath}
\mathbf{T} = \mathbf{X}_1\mathbf{X}_2^T.
\end{displaymath}
When $\mathbf{X}_1$ and $\mathbf{X}_2$ are of rank $R$ then the above expression is a full rank factorization of $\mathbf{T}$. If $\mathbf{T} = \mathbf{X}_1^{\prime}{\mathbf{X}_2^{\prime}}^T$ is another full rank factorization of $\mathbf{T}$, then there exists an invertible matrix $\mathbf{W}\in\mathbb{R}^{R\times R}$ such that $\mathbf{X}_1^{\prime} = \mathbf{X}_1\mathbf{W}$ and $\mathbf{X}_2^{\prime} = \mathbf{X}_2(\mathbf{W}^{-1})^T$ due to \cite[Theorem 2]{piziak1999full}. Since the alphabets in our study are finite, the number of realizations of the principal minor $\{\mathbf{X}_1\}_{1:R,1:R}$ is finite. Therefore the number of such matrices $\mathbf{W}$ is finite and does not depend on $n$. We denote the corresponding upper bound on the number of matrices $\mathbf{W}$ by $\Gamma_2$.
\end{proof}
By combining Lemma \ref{lemma_1} and Lemma \ref{lemm_2} our result when the order of tensor $N\geq 2$ and the number of components $R\geq 2$ is given in the following.
\begin{theorem}\label{thm_2}
The minimum almost-lossless compression threshold is given by
\begin{displaymath}
C^{\star}_t = \sum_{i,r}H(P_{\mathcal{X}_i,r}).
\end{displaymath}
\end{theorem}
We now describe a full compression scheme similar to the one in Theorem \ref{thm_1}. 
\begin{proof}
For each mode $i\in [1:N]$ we define a typical set 
\begin{displaymath}
\mathcal{A}_{i,\gamma}^n = \big\{\mathbf{X}_i\mid |-\log P(\mathbf{X}_i)-n\sum_r H(P_{\mathcal{X}_i,r})|<n\gamma\big\},
\end{displaymath}
where for each $i\in [1:N]$ the probability of a realization of factor matrix $\mathbf{X}_i$, $P(\mathbf{X}_i)$, is given by $P(\mathbf{X}_i) = \prod_{r=1}^R P_{\mathcal{X}_i,r}^n(\mathbf{a}_{ir})$.
Then we define two sets $\mathcal{S}_{n,\gamma} = \bigtimes_{i=1}^N \mathcal{A}_{i,\gamma}^n$ and 
\begin{displaymath}
\mathcal{T}_{n,\gamma} = \{[\mathbf{X}_1;\cdots;\mathbf{X}_N]\mid (\mathbf{X}_1,\dots,\mathbf{X}_N) \in \mathcal{S}_{n,\gamma}\}.
\end{displaymath}
We index all elements in $\mathcal{T}_{n,\gamma}$. If $\mathbf{T}\in \mathcal{T}_{n,\gamma}$ we store the corresponding index, otherwise we store a given index. We need at most $n(\sum_{i,r} H(P_{\mathcal{X}_i,r})+\eta)$ nats for the indexing scheme. Hence we have $C^{\star}_t\leq \sum_{i,r} H(P_{\mathcal{X}_i,r})$.

In the converse direction, let $(\phi_n,\psi_n)$ be a sequence of tensor compression-reconstruction mappings such that the compression threshold $C$ is almost-losslessly achievable. For a given $\epsilon>0$ there exists a sufficiently large $n_0(\epsilon)$ such that
\begin{displaymath}
\mathrm{Pr}\{\psi_n(\phi_n(\mathbf{T}))\neq \mathbf{T}\}\leq \epsilon,\;\forall n\geq n_0(\epsilon).
\end{displaymath}
Similarly we define the decodable set of tensors
\begin{displaymath}
\mathcal{D}_n = \{\mathbf{T}\mid \phi_n(\psi_n(\mathbf{T})) = \mathbf{T}\},
\end{displaymath}
and the expanding set of corresponding tuples of factor matrices by
\begin{displaymath}
\mathcal{S}_n = \{(\mathbf{X}_i)_{i=1}^N\mid [\mathbf{X}_1;\cdots;\mathbf{X}_N] \in \mathcal{D}_n\}.
\end{displaymath}
We also define the set of tuples of full rank factors
\begin{displaymath}
\mathcal{E}_n = \{(\mathbf{X}_i)_{i=1}^N\mid \mathrm{rank}(\mathbf{X}_i) = R,\;\forall i\in [1:N]\}.
\end{displaymath}
By our previous analysis we have $|\mathcal{D}_n|\leq |\mathcal{M}|$ and $|\mathcal{S}_n\cap \mathcal{E}_n|\stackrel{\text{Lemma \ref{lemm_2}}}{\leq} \Gamma_N|\mathcal{D}_n|\leq \Gamma_N |\mathcal{M}|$ for some large enough constant $\Gamma_N$. We define an atypical set
\begin{displaymath}
\mathcal{T}_{n,\eta} = \{(\mathbf{X}_i)_{i=1}^N \mid -\sum_{i=1}^N\log P(\mathbf{X}_i)\geq \log|\mathcal{M}|+ \eta\}.
\end{displaymath}
We then have
\begin{align}
\mathrm{Pr}\{&(\mathbf{X}_i)_{i=1}^N \in \mathcal{T}_{n,\eta}\} = \mathrm{Pr}\{(\mathbf{X}_i)_{i=1}^N \in \mathcal{S}_n\cap \mathcal{E}_n\cap \mathcal{T}_{n,\eta}\}+ \mathrm{Pr}\{(\mathbf{X}_i)_{i=1}^N \in \big(\mathcal{S}_n\cap \mathcal{E}_n\big)^c\cap \mathcal{T}_{n,\eta}\}\nonumber\\
&\leq \mathrm{Pr}\{(\mathbf{X}_i)_{i=1}^N \in \mathcal{S}_n^c\} + \sum_{i}\mathrm{Pr}\{\mathrm{rank}(\mathbf{X}_i) < R\}  + \mathrm{Pr}\{(\mathbf{X}_i)_{i=1}^N \in \mathcal{S}_n\cap \mathcal{E}_n\cap \mathcal{T}_{n,\eta}\}\nonumber\\
& =\mathrm{Pr}\{\psi_n(\phi_n(\mathbf{T}))\neq \mathbf{T}\} + \sum_{i}\mathrm{Pr}\{\mathrm{rank}(\mathbf{X}_i) < R\}  + \sum_{(\mathbf{X}_i)_{i=1}^N\in \mathcal{S}_n\cap \mathcal{E}_n\cap \mathcal{T}_{n,\eta}}\prod_{i=1}^N P(\mathbf{X}_i)\nonumber\\
&\leq \epsilon + \sum_{i}\mathrm{Pr}\{\mathrm{rank}(\mathbf{X}_i) < R\} + |\mathcal{S}_n\cap \mathcal{E}_n\cap \mathcal{T}_{n,\eta}|e^{-\eta}/|\mathcal{M}|\nonumber\\
&\stackrel{\text{Lemma \ref{lemm_2}}}{\leq} \epsilon + \sum_{i}\mathrm{Pr}\{\mathrm{rank}(\mathbf{X}_i) < R\} + e^{-\eta}\Gamma_N.
\end{align}
We select $\eta = n\gamma$ where $\gamma>0$ is an arbitrary number and apply a similar line of reasoning as in the converse proof of Theorem \ref{thm_1}. If for all $n_1\geq n_0(\epsilon)$ there exists a $n>n_1$ such that $\log |\mathcal{M}| + \eta<\break n(\sum_{i,r}H(P_{\mathcal{X}_i,r})- NR\gamma)$, then we have 
\begin{displaymath}
\limsup_{n\to\infty} \mathrm{Pr}\{(\mathbf{X}_i)_{i=1}^N \in \mathcal{T}_{n,\eta}\} = 1,
\end{displaymath}
due to the weak law of large numbers. The last inequality is violated since by Lemma \ref{lemma_1} we have
\begin{displaymath}\sum_{i}\mathrm{Pr}\{\mathrm{rank}(\mathbf{X}_i) < R\}\to 0,\;\text{as}\; n\to\infty.
\end{displaymath}
 Therefore we must have $\log |\mathcal{M}| + n\gamma \geq n (\sum_{i,r}H(P_{\mathcal{X}_i,r}) - NR\gamma)$ for all $n\geq n_{\epsilon}$. Since $\gamma$ and $\epsilon$ are arbitrary we have $C^{\star}_t\geq \sum_{i,r}H(P_{\mathcal{X}_i,r})$.
\end{proof}

\bibliographystyle{siamplain}
\bibliography{references}

\begin{thebibliography}{10}

\bibitem{carroll1970analysis}
{\sc J.~D. Carroll and J.-J. Chang}, {\em {Analysis of individual differences
  in multidimensional scaling via an N-way generalization of “Eckart-Young”
  decomposition}}, Psychometrika, 35 (1970), pp.~283--319.

\bibitem{cichocki2015tensor}
{\sc A.~Cichocki, D.~Mandic, L.~De~Lathauwer, G.~Zhou, Q.~Zhao, C.~Caiafa, and
  H.~A. Phan}, {\em Tensor decompositions for signal processing applications:
  {F}rom two-way to multiway component analysis}, {IEEE Signal Processing
  Magazine}, 32 (2015), pp.~145--163.

\bibitem{cover1999elements}
{\sc T.~M. Cover}, {\em Elements of information theory}, John Wiley \& Sons,
  1999.

\bibitem{hanspectrum}
{\sc T.~S. Han}, {\em Information-Spectrum Methods in Information Theory},
  Springer-Verlag Berlin Heidelberg, 2003.

\bibitem{harshman1970foundations}
{\sc R.~A. Harshman et~al.}, {\em {Foundations of the PARAFAC procedure: Models
  and conditions for an ``explanatory" multimodal factor analysis}},  (1970).

\bibitem{hillar2013most}
{\sc C.~J. Hillar and L.-H. Lim}, {\em Most tensor problems are {NP-hard}},
  Journal of the ACM (JACM), 60 (2013), pp.~1--39.

\bibitem{kolda2009tensor}
{\sc T.~G. Kolda and B.~W. Bader}, {\em Tensor decompositions and
  applications}, SIAM review, 51 (2009), pp.~455--500.

\bibitem{komlos1968determinant}
{\sc J.~Koml{\'o}s}, {\em On the determinant of random matrices}, Studia
  Scientiarum Mathematicarum Hungarica, 3 (1968), pp.~387--399.

\bibitem{kruskal1977three}
{\sc J.~B. Kruskal}, {\em Three-way arrays: rank and uniqueness of trilinear
  decompositions, with application to arithmetic complexity and statistics},
  Linear algebra and its applications, 18 (1977), pp.~95--138.

\bibitem{piziak1999full}
{\sc R.~Piziak and P.~Odell}, {\em Full rank factorization of matrices},
  {Mathematics Magazine}, 72 (1999), pp.~193--201.

\bibitem{sidiropoulos2000uniqueness}
{\sc N.~D. Sidiropoulos and R.~Bro}, {\em On the uniqueness of multilinear
  decomposition of {N}-way arrays}, Journal of Chemometrics: A Journal of the
  Chemometrics Society, 14 (2000), pp.~229--239.

\bibitem{stegeman2007kruskal}
{\sc A.~Stegeman and N.~D. Sidiropoulos}, {\em {On Kruskal’s uniqueness
  condition for the Candecomp/Parafac decomposition}}, {Linear Algebra and its
  Applications}, 420 (2007), pp.~540--552.

\bibitem{tao2006random}
{\sc T.~Tao and V.~Vu}, {\em On random $\pm$1 matrices: Singularity and
  determinant}, Random Structures \& Algorithms, 28 (2006), pp.~1--23.

\bibitem{tucker1963}
{\sc L.~R. Tucker}, {\em Implications of factor analysis of three way matrices
  for measurements of change}, In Harris, CW (Editor), Problems in measuring
  change,  (1963), pp.~122--137.

\end{thebibliography}

\end{document}